\documentclass[12pt]{article}%
\usepackage{amsfonts}
\usepackage{amsmath}
\usepackage{amssymb}
\usepackage{graphicx}%
\newtheorem{theorem}{Theorem}
\newtheorem{acknowledgement}[theorem]{Acknowledgement}

\newtheorem{corollary}[theorem]{Corollary}

\newtheorem{lemma}[theorem]{Lemma}

\newtheorem{proposition}[theorem]{Proposition}
\newtheorem{remark}[theorem]{Remark}

\newenvironment{proof}[1][Proof]{\noindent\textbf{#1.} }{\ \rule{0.5em}{0.5em}}
\begin{document}

\title{A new approach to the $\star$-genvalue equation}
\author{Maurice de Gosson\thanks{This author has been financed by the Austrian
Research Agency FWF (Projekt \textquotedblleft Symplectic Geometry and
Applications to TFA and QM\textquotedblright, Projektnummer P20442-N13). } \ ,
\ Franz Luef\thanks{This author has been supported by the European Union
EUCETIFA grant MEXT-CT-2004-517154.}\\\textit{Universit\"{a}t Wien}\\\textit{Fakult\"{a}t f\"{u}r Mathematik, NuHAG }\\\textit{Nordbergstrasse 15, AT-1090 Wien}}
\maketitle

\begin{abstract}
We show that the eigenvalues and eigenfunctions of the stargenvalue equation
can be completely expressed in terms of the corresponding eigenvalue problem
for the quantum Hamiltonian. Our method makes use of a Weyl-type
representation of the star-product and of the properties of the cross-Wigner
transform, which appears as an intertwining operator.

\end{abstract}

\textbf{MSC\ (2000)}: 47G30, 81S10

\textbf{Keywords:} Moyal product, stargenvalue equation, Wigner transform,
Weyl operator

\section{Introduction and Motivation}

One of the key equations in the deformation quantization theory of Bayen et
al. \cite{BFFLS1,BFFLS2} is, no doubt, the stargenvalue (for short $\star
$-genvalue) equation $H\star\Psi=E\Psi$ where $\star$ is the Moyal--Groenewold
\textquotedblleft star-product\textquotedblright\cite{BFFLS1,BFFLS2,groen}. In
this Letter we show that the $\star$-genvalue equation can be completely
solved in terms of the usual eigenvalue/eigenfunction problem $\widehat{H}%
\psi=E\psi$ where $\widehat{H}$ is the Weyl operator with symbol $H$ (and vice
versa). The underlying idea is simple: we first rewrite the equation
$H\star\Psi=E\Psi$ in the form
\[
H(x+\tfrac{1}{2}i\hbar\partial_{p},p-\tfrac{1}{2}i\hbar\partial_{x}%
)\Psi(x,p)=E\Psi(x,p)\text{,}%
\]
where $H(x+\tfrac{1}{2}i\hbar\partial_{p},p-\tfrac{1}{2}i\hbar\partial_{x})$
is the Weyl operator with symbol
\[
\mathbb{H}(z,\zeta)=H(x-\tfrac{1}{2}\zeta_{p},p+\tfrac{1}{2}\zeta_{x}).
\]
We next show that the solutions of this equation and those of $\widehat{H}%
\psi=E\psi$ can be obtained from each other using a family of intertwining
operators (which is countable when $\widehat{H}$ is essentially self-adjoin);
these operators are up to a normalization factor, the cross-Wigner transforms
$\psi\longmapsto W(\psi,\phi)$ where $\phi$ describes the set of
eigenfunctions of $\widehat{H}$. Our approach is inspired by previous work
\cite{jphysa} of one of us on the time-dependent Torres-Vega \cite{TV}
Schr\"{o}dinger equation in phase space.

\subsubsection*{Notation}

We will write $z=(x,p)$ where $x\in\mathbb{R}^{n}$ and $p\in(\mathbb{R}%
^{n})^{\ast}$. Operators $\mathcal{S}(\mathbb{R}^{n})\longrightarrow$
$\mathcal{S}^{\prime}(\mathbb{R}^{n})$ are usually denoted by $\widehat
{A},\widehat{B},...$ while operators $\mathcal{S}(\mathbb{R}^{2n}%
)\longrightarrow\mathcal{S}^{\prime}(\mathbb{R}^{2n})$ are denoted by
$\widetilde{A},\widetilde{B},...$ The Greek letters $\psi,\phi,...$ stand for
functions defined on $\mathbb{R}^{n}$ while their capitalized counterparts
$\Psi,\Phi,...$ denote functions defined on $\mathbb{R}^{2n}$. We will make
use of the symplectic Fourier transform defined for $\Psi\in\mathcal{S}%
(\mathbb{R}^{2n})$ by the formula%
\[
\Psi_{\sigma}^{\hbar}(z)=F_{\sigma}^{\hbar}\Psi(z)=\left(  \tfrac{1}{2\pi
\hbar}\right)  ^{n}\int_{\mathbb{R}^{2n}}e^{-\frac{i}{\hbar}\sigma
(z,z^{\prime})}\Psi(z^{\prime})dz^{\prime}%
\]
where $\sigma(z,z^{\prime})=p\cdot x^{\prime}-p^{\prime}\cdot x$ is the
standard symplectic form on $\mathbb{R}^{n}\times(\mathbb{R}^{n})^{\ast}%
\equiv\mathbb{R}^{2n}$ (the dot $\cdot$ stands for the duality bracket; in
practice $p\cdot x$ can be seen as the usual Euclidean scalar product under
the identification $(\mathbb{R}^{n})^{\ast}\equiv\mathbb{R}^{n}$). The
symplectic Fourier transform is involutive: $F_{\sigma}^{\hbar}\circ
F_{\sigma}^{\hbar}$ is the identity on $\mathcal{S}^{\prime}(\mathbb{R}^{2n})$.

\section{Stargenvalue Equation:\ Short Review}

In view of Schwartz's kernel theorem every linear continuous operator
$\widehat{A}:\mathcal{S}(\mathbb{R}^{n})\longrightarrow\mathcal{S}^{\prime
}(\mathbb{R}^{n})$ can be represented, for $\psi\in\mathcal{S}(\mathbb{R}%
^{n})$, in the form $\widehat{A}\psi(x)=\left\langle \mathcal{K}_{A}%
(x,\cdot),\psi\right\rangle $ with $\mathcal{K}_{A}\in\mathcal{S}^{\prime
}(\mathbb{R}^{n}\times\mathbb{R}^{n})$. By definition the contravariant (Weyl)
symbol of $\widehat{A}$ is the tempered distribution $A$ defined by the
Fourier transform%
\begin{equation}
a(x,p)=\left\langle e^{-\frac{i}{\hbar}p(\cdot)},\mathcal{K}_{\widehat{A}%
}(x+\tfrac{1}{2}(\cdot),x-\tfrac{1}{2}(\cdot)\right\rangle . \label{axy}%
\end{equation}
Assume that $\widehat{B}:\mathcal{S}(\mathbb{R}^{n})\longrightarrow
\mathcal{S}(\mathbb{R}^{n})$; then the product $\widehat{C}=\widehat{A}%
\circ\widehat{B}$ exists and its Weyl symbol is given by the Moyal product
\begin{equation}
a\star b(z)=\left(  \tfrac{1}{4\pi\hbar}\right)  ^{2n}\iint
\nolimits_{\mathbb{R}^{n}\times\mathbb{R}^{n}}e^{\frac{i}{2\hbar}\sigma
(u,v)}a(z+\tfrac{1}{2}u)b(z-\tfrac{1}{2}v)dudv. \label{cz}%
\end{equation}
The main observation we will exploit in this paper is the following: if we
write $a=H$ and $b=\Psi$ then we can write%
\begin{equation}
H\star\Psi=\widetilde{H}\Psi\label{hstarpsi},%
\end{equation}
where%
\[
\widetilde{H}=H(x+\tfrac{1}{2}i\hbar\partial_{p},p-\tfrac{1}{2}i\hbar
\partial_{x})
\]
is a certain pseudodifferential operator on $\mathcal{S}(\mathbb{R}^{2n})$ we
are going to identify. Let us view the linear operator $\widetilde{H}%
:\Psi\longmapsto H\star\Psi$ on $\mathcal{S}(\mathbb{R}^{2n})$ as a Weyl
operator. Using formula (\ref{cz}), the kernel of $\widetilde{H}$ is the
distribution%
\begin{equation}
\mathcal{K}_{\widetilde{H}}(z,y)=\left(  \tfrac{1}{2\pi\hbar}\right)
^{2n}\int_{\mathbb{R}^{2n}}e^{\frac{i}{\hbar}\sigma(u,z-y)}H(z-\tfrac{1}%
{2}u)du \label{k},%
\end{equation}
hence using (\ref{axy}) and the Fourier inversion formula the contravariant
symbol of $\widetilde{H}$ is%
\[
\mathbb{H}(z,\zeta)=\int_{\mathbb{R}^{2n}}e^{\frac{i}{\hbar}\zeta\cdot\eta
}\mathcal{K}_{\widetilde{H}}(z+\tfrac{1}{2}\eta,z-\tfrac{1}{2}\eta)d\eta.
\]
Using (\ref{k}) and performing the change of variables $u=2z+\eta-z^{\prime}$
we get
\[
\mathcal{K}_{\widetilde{H}}(z+\tfrac{1}{2}\eta,z-\tfrac{1}{2}\eta)=\left(
\tfrac{1}{2\pi\hbar}\right)  ^{2n}e^{\frac{2i}{\hbar}\sigma(z,\eta)}%
\int_{\mathbb{R}^{2n}}e^{\frac{i}{\hbar}\sigma(\eta,z^{\prime})}H(\tfrac{1}%
{2}z^{\prime})dz^{\prime};
\]
setting $H(\tfrac{1}{2}z^{\prime})=H_{1/2}(z^{\prime})$ the integral is
$\left(  2\pi\hbar\right)  ^{n}$ times the symplectic Fourier transform
$F_{\sigma}^{\hbar}H_{1/2}(-\eta)=(H_{1/2})_{\sigma}(-\eta)$ so that%
\begin{align*}
\mathbb{H}(\tfrac{1}{2}z,\zeta)  &  =\left(  \tfrac{1}{2\pi\hbar}\right)
^{n}\int_{\mathbb{R}^{2n}}e^{\frac{i}{\hbar}\zeta\cdot\eta}e^{\frac{i}{\hbar
}\sigma(z,\eta)}(H_{1/2})_{\sigma}(-\eta)d\eta\\
&  =\left(  \tfrac{1}{2\pi\hbar}\right)  ^{n}\int_{\mathbb{R}^{2n}}%
e^{-\frac{i}{\hbar}\sigma(z-J\zeta,\eta)}(H_{1/2})_{\sigma}(\eta)d\eta
\end{align*}
where $J=%
\begin{pmatrix}
0 & I\\
-I & 0
\end{pmatrix}
$ is the standard symplectic matrix. Since the second equality is the inverse
symplectic Fourier transform of $(H_{1/2})_{\sigma}$ calculated at the point
$z+J\zeta$. We finally get%
\begin{equation}
\mathbb{H}(z,\zeta)=H(x-\tfrac{1}{2}\zeta_{p},p+\tfrac{1}{2}\zeta_{x})
\label{hsymb}%
\end{equation}
where we are viewing $\zeta=(\zeta_{x},\zeta_{p})$ as the dual variable of
$z=(x,p)$; this justifies formula (\ref{hstarpsi}) viewing $\widetilde{H}$ as
the quantized Hamiltonian obtained from $\mathbb{H}$ by the quantum rule%
\begin{equation}
(x,p)\longmapsto(x+\tfrac{1}{2}i\hbar\partial_{p},p-\tfrac{1}{2}i\hbar
\partial_{x}). \label{qr}%
\end{equation}

\section{$\widetilde{H}$-Calculus}

There is another very fruitful way of interpreting the Weyl operators
$\widetilde{H}=H\star\Psi$. Let us return to the expression (\ref{cz}) with
$a=H$ and $b=\Psi;$ performing the changes of variable $u=2(z^{\prime}-z)$ and
$v=z_{0}$ this formula can be rewritten as%
\[
\widetilde{H}\Psi(z)=\left(  \tfrac{1}{2\pi\hbar}\right)  ^{2n}\int
_{\mathbb{R}^{2n}}\left[  \int_{\mathbb{R}^{2n}}e^{-\frac{i}{\hbar}%
\sigma(z_{0},z^{\prime})}H(z^{\prime})dz^{\prime}\right]  e^{-\frac{i}{\hbar
}\sigma(z,z_{0})}\Psi(z-\tfrac{1}{2}z_{0})dz_{0}.
\]
Observing that the integral between brackets is $(2\pi\hbar)^{n}$ times the
symplectic Fourier transform of $H$ we can write this formula in the form%
\begin{equation}
\widetilde{H}\Psi(z)=\left(  \tfrac{1}{2\pi\hbar}\right)  ^{n}\int
_{\mathbb{R}^{2n}}H_{\sigma}^{\hbar}(z_{0})\widetilde{T}(z_{0})\Psi(z)dz_{0}
\label{psi1}%
\end{equation}
where $\widetilde{T}(z_{0})$ is the operator defined by
\begin{equation}
\widetilde{T}(z_{0})\Psi(z)=e^{-\frac{i}{\hbar}\sigma(z,z_{0})}\Psi
(z-\tfrac{1}{2}z_{0}). \label{hwbis}%
\end{equation}
Formula (\ref{psi1}) is strongly reminiscent of the representation
\begin{equation}
\widehat{H}\psi=\left(  \tfrac{1}{2\pi\hbar}\right)  ^{n}\int_{\mathbb{R}%
^{2n}}H_{\sigma}^{\hbar}(z_{0})\widehat{T}(z_{0})\psi dz_{0} \label{weyl}%
\end{equation}
of a Weyl operator $\widehat{H}$ in terms of its covariant symbol $H_{\sigma
}^{\hbar}=F_{\sigma}^{\hbar}H$ and the Heisenberg--Weyl operator%
\[
\widehat{T}(z_{0})\psi(x)=e^{\frac{i}{\hbar}(p_{0}\cdot x-\frac{1}{2}%
p_{0}\cdot x_{0})}\psi(x-x_{0})\text{,}%
\]
except that $\widetilde{T}(z_{0})$ is allowed to act on functions of $z$ and
not only of $x$. This feeling is amplified when one notes after a
straightforward calculation that the operators $\widetilde{T}(z_{0})$ obey the
relations%
\begin{align}
\widetilde{T}(z_{0}+z_{1})  &  =e^{-\frac{i}{2\hbar}\sigma(z_{0},z_{1}%
)}\widetilde{T}(z_{0})\widetilde{T}(z_{1})\label{a}\\
\widetilde{T}(z_{1})\widetilde{T}(z_{0})  &  =e^{-\frac{i}{\hbar}\sigma
(z_{0},z_{1})}\widetilde{T}(z_{0})\widetilde{T}(z_{1}) \label{b}%
\end{align}
which are similar to those satisfied by the Heisenberg--Weyl operators. These
facts suggest that $\widetilde{T}(z_{0},t)=e^{\frac{it}{\hbar}}\widetilde
{T}(z_{0})$ defines a unitary representation of the Heisenberg group. Let us
prove this is indeed the case. For this we will need the linear mapping
$W_{\phi}:\mathcal{S}(\mathbb{R}^{n})\longrightarrow S(\mathbb{R}^{2n})$
defined by
\begin{equation}
W_{\phi}\psi=(2\pi\hbar)^{n/2}W(\psi,\phi) \label{wpf}%
\end{equation}
where $\phi$ denotes an arbitrary function in $\mathcal{S}(\mathbb{R}^{n})$
such that $||\phi||_{L^{2}}=1$. $W(\psi,\phi)$ is the cross-Wigner
distribution; we thus have explicitly%
\begin{equation}
W_{\phi}\psi(z)=\left(  \tfrac{1}{2\pi\hbar}\right)  ^{n/2}\int_{\mathbb{R}%
^{n}}e^{-\frac{i}{\hbar}p\cdot y}\psi(x+\tfrac{1}{2}y)\overline{\phi
(x-\tfrac{1}{2}y)}dy. \label{wifi}%
\end{equation}
In view of Moyal's identity%
\[
(W(\psi,\phi)|W(\psi^{\prime},\phi^{\prime}))_{L^{2}(\mathbb{R}^{2n})}=\left(
\tfrac{1}{2\pi\hbar}\right)  ^{n}(\psi|\psi^{\prime})_{L^{2}(\mathbb{R}^{n}%
)}(\phi|\phi^{\prime})_{L^{2}(\mathbb{R}^{n})}%
\]
the operator $W_{\phi}$ extends into an isometry of $L^{2}(\mathbb{R}^{n})$
onto a subspace $\mathcal{H}_{\phi}$ of $L^{2}(\mathbb{R}^{2n})$; we are going
to see in a moment $\mathcal{H}_{\phi}$ is closed in $L^{2}(\mathbb{R}^{2n})$,
but let us first give a formula for the adjoint $W_{\phi}^{\ast}$ of $W_{\phi
}$. We have%
\begin{equation}
W_{\phi}^{\ast}\Psi(z)=\left(  \tfrac{2}{\pi\hbar}\right)  ^{n/2}%
\int_{\mathbb{R}^{n}}e^{\frac{2i}{\hbar}p\cdot(x-y)}\phi(2y-x)\Psi(y,p)dpdy
\label{adj}%
\end{equation}
(it follows from a straightforward calculation using the identity $(W_{\phi
}\psi|\Psi)_{L^{2}(\mathbb{R}^{2n})}=(\psi|W_{\phi}^{\ast}\Psi)_{L^{2}%
(\mathbb{R}^{n})}$).

\begin{proposition}
The range $\mathcal{H}_{\phi}$ of $W_{\phi}$ is closed, and hence a Hilbert space.
\end{proposition}

\begin{proof}
Set $P_{\phi}=W_{\phi}W_{\phi}^{\ast}$ where $W_{\phi}^{\ast}$ is the adjoint
of $W_{\phi}$; we have $P_{\phi}=P_{\phi}^{\ast}$ and $P_{\phi}P_{\phi}^{\ast
}=P_{\phi}$ hence $P_{\phi}$ is an orthogonal projection. Since $W_{\phi
}^{\ast}W_{\phi}$ is the identity on $L^{2}(\mathbb{R}^{n})$ the range of
$W_{\phi}^{\ast}$ is $L^{2}(\mathbb{R}^{n})$ and that of $P_{\phi}$ is
therefore precisely $\mathcal{H}_{\phi}$. Since the range of a projection is
closed, so is $\mathcal{H}_{\phi}$.
\end{proof}

This result, together with formula (\ref{b}) shows that $\widetilde{T}(z_{0})$
and $\widehat{T}(z_{0})$ are unitarily equivalent representations of the
Heisenberg group $\mathbf{H}_{n}$; the irreducibility of the representation
$\widetilde{T}(z_{0}):\mathbf{H}_{n}\longrightarrow\mathcal{H}_{\phi}$ follows
from von Neumann's uniqueness theorem for the projective representations of
the CCR.

Let us return to the operator $\widetilde{H}=H\star$. A straightforward
calculation showing that $W_{\phi}$ satisfies the intertwining relations%
\begin{align*}
x\star W_{\phi}\psi &  =(x+\tfrac{1}{2}i\hbar\partial_{p})W_{\phi}\psi
=W_{\phi}(x\psi)\\
p\star W_{\phi}\psi &  =(p-\tfrac{1}{2}i\hbar\partial_{x})W_{\phi}\psi
=W_{\phi}(-i\hbar\partial_{x}\psi)
\end{align*}
an educated guess is then that more generally:

\begin{proposition}
(i) The operator $W_{\phi}$ intertwines the operators $\widetilde{T}(z_{0})$
and $\widehat{T}(z_{0})$:
\begin{equation}
W_{\phi}(\widehat{T}(z_{0})\psi)=\widetilde{T}(z_{0})W_{\phi}\psi;
\label{inter1}%
\end{equation}
(ii) We also have
\begin{equation}
\widetilde{H}W_{\phi}=W_{\phi}\widehat{H}\text{ \ and \ }W_{\phi}^{\ast
}\widetilde{H}=\widehat{H}W_{\phi}^{\ast}\text{.} \label{fund}%
\end{equation}

\end{proposition}

\begin{proof}
Making the change of variable $y=y^{\prime}+x_{0}$ in the definition
(\ref{wifi}) of $W_{\phi}$ we get
\[
W_{\phi}(\widehat{T}(z_{0})\psi,\phi)(z)=e^{-\frac{i}{\hbar}\sigma(z,z_{0}%
)}W_{\phi}\psi(z-\tfrac{1}{2}z_{0})
\]
which is precisely (\ref{inter1}). Applying $W_{\phi}$ to both sides of
(\ref{weyl}), we get
\[
W_{\phi}\widehat{H}\psi=\left(  \tfrac{1}{2\pi\hbar}\right)  ^{n}%
\int_{\mathbb{R}^{2n}}H_{\sigma}^{\hbar}(z_{0})W_{\phi}[\widehat{T}(z_{0}%
)\psi]dz_{0}.
\]
and hence%
\[
W_{\phi}\widehat{H}\psi=\left(  \tfrac{1}{2\pi\hbar}\right)  ^{n}%
\int_{\mathbb{R}^{2n}}H_{\sigma}^{\hbar}(z_{0})[\widetilde{T}(z_{0})W_{\phi
}\psi]dz_{0},%
\]
which is the first equality (\ref{fund}) in view of formula (\ref{psi1}). To
prove the second equality it suffices to apply the first to $W_{\phi}^{\ast
}\widetilde{H}=(\widetilde{H}^{\ast}W_{\phi})^{\ast}$.
\end{proof}

\section{Spectral Results}

We will need the following result, which is quite interesting by itself:

\begin{lemma}
\label{Wong}Let $(\phi_{j})_{j}$ be an arbitrary orthonormal basis of
$L^{2}(\mathbb{R}^{n})$. Then the vectors $\Phi_{j,k}=W_{\phi_{j}}\phi_{k}$ form an
orthonormal basis of $L^{2}(\mathbb{R}^{2n})$.
\end{lemma}

\begin{proof}
Since the $W_{\phi_{j}}$ are isometries the vectors $\Phi_{j,k}$ form an
orthonormal system. It is sufficient to show that if $\Psi\in L^{2}%
(\mathbb{R}^{2n})$ is orthogonal to the family $(\Phi_{j,k})_{j,k}$ (and hence
to all the spaces $\mathcal{H}_{\phi_{j}}$) then $\Psi=0$. Assume that
$(\Psi|\Phi_{jk})_{L^{2}(\mathbb{R}^{2n})}=0$ for all $j,k.$ Since we have
\[
(\Psi|\Phi_{jk})_{L^{2}(\mathbb{R}^{2n})}=(\Psi|W_{\phi_{j}}\phi_{k}%
)_{L^{2}(\mathbb{R}^{2n})}=(W_{\phi_{j}}^{\ast}\Psi|\phi_{k})_{L^{2}%
(\mathbb{R}^{n})}%
\]
it follows that $W_{\phi_{j}}^{\ast}\Psi=0$ for all $j$ since $(\phi_{j})_{j}$
is a basis; using the anti-linearity of $W_{\phi}$ in $\phi$ we have in fact
$W_{\phi}^{\ast}\Psi=0$ for all $\phi\in L^{2}(\mathbb{R}^{n})$. Let us show
that this implies that $\Psi=0$. In view of formula (\ref{adj}) for the
adjoint of $W_{\phi}$ the operator $W_{\phi}^{\ast}$ has kernel%
\[
\Phi_{x}(y,p)=\left(  \tfrac{2}{\pi\hbar}\right)  ^{n/2}e^{\frac{2i}{\hbar
}p\cdot(x-y)}\phi(2y-x).
\]
Let us fix $x$; the property $W_{\phi}^{\ast}\Psi=0$ for all $\phi$ is then
equivalent to $\left\langle \Psi,\Phi_{x}\right\rangle =0$ for all $\Phi
_{x}\in\mathcal{S}(\mathbb{R}^{2n})$ (fixed $x$) and hence $\Psi=0$, which we
set out to show.
\end{proof}

We now have everything we need to prove the main results of this Letter. We
begin by stating the following general property:

\begin{theorem}
\label{eigen1}The following properties are true: (i) The eigenvalues of the
operators $\widehat{H}$ and $\widetilde{H}=H\star$ are the same; (ii) Let
$\psi$ be an eigenfunction of $\widehat{H}$: $\widehat{H}\psi=\lambda\psi$.
Then, for every $\phi$, the function $\Psi=W_{\phi}\psi$ is an eigenfunction
of $\widetilde{H}$ corresponding to the same eigenvalue: $\widetilde{H}%
\Psi=\lambda\Psi$. (iii) Conversely, if $\Psi$ is an eigenfunction of
$\widetilde{H}$ then $\psi=W_{\phi}^{\ast}\Psi$ is an eigenfunction of
$\widehat{H}$ corresponding to the same eigenvalue.
\end{theorem}

\begin{proof}
That every eigenvalue of $\widehat{H}$ also is an eigenvalue of $\widetilde
{H}$ is clear: if $\widehat{H}\psi=\lambda\psi$ for some $\psi\neq0$ then
\[
\widetilde{H}(W_{\phi}\psi)=W_{\phi}\widehat{H}\psi=\lambda(W_{\phi}\psi)
\]
and $W_{\phi}\psi\neq0$ because $W_{\phi}$ is injective; this proves at the
same time that $W_{\phi}\psi$ is an eigenfunction of $\widetilde{H}$. Assume
conversely that $\widetilde{H}\Psi=\lambda\Psi$ for $\Psi\neq0$ and
$\lambda\in\mathbb{R}$. For every $\phi$ we have, using the second equality
(\ref{fund}),
\[
\widehat{H}W_{\phi}^{\ast}\Psi=W_{\phi}^{\ast}\widetilde{H}\Psi=\lambda
W_{\phi}^{\ast}\Psi
\]
hence $\lambda$ is an eigenvalue of $\widehat{H}$; $W_{\phi}^{\ast}\Psi$ is an
an eigenfunction if it is different from zero. Let us prove this is indeed the
case. We have $W_{\phi}W_{\phi}^{\ast}\Psi=P_{\phi}\Psi$ where $P_{\phi}$ is
the orthogonal projection on the range $\mathcal{H}_{\phi}$ of $W_{\phi}$.
Assume that $W_{\phi}^{\ast}\Psi=0$; then $P_{\phi}\Psi=0$ for every $\phi
\in\mathcal{S}(\mathbb{R}^{n})$, and hence $\Psi=0$ in view of Lemma
\ref{Wong} above.
\end{proof}

\begin{remark}
The result above is indeed quite general, because we do not make any
assumption on the multiplicity of the (star)eigenvalues, nor do we assume that
$\widehat{H}$ is essentially self-adjoint. Notice that the proof actually
works for arbitrary $\phi\in\mathcal{S}^{\prime}(\mathbb{R}^{n})$. We present some 
examples at the end of this section.
\end{remark}

\begin{corollary}
\label{kernel}Suppose that $\widehat{H}$ is an essentially self-adjoint
operator on $L^{2}(\mathbb{R}^{n})$ and that each of the eigenvalues
$\lambda_{0},\lambda_{1},...,\lambda_{j},...$ has multiplicity one. Let
$\psi_{0},\psi_{1},...,\psi_{j},...$ be a corresponding sequence of
orthonormal eigenfunctions. Let $\Psi_{j}$ be an eigenfunction of
$\widetilde{H}$ corresponding to the eigenvalue $\lambda_{j}$. Then there exists a
sequence $(\alpha_{j,k})_{k}$ of complex numbers such that
\begin{equation}
\Psi_{j}=\sum_{\ell}\alpha_{j,\ell}\Psi_{j,\ell}\text{ \ with \ }\Psi_{j,\ell
}=W_{\psi_{\ell}}\psi_{j}\in\mathcal{H}_{j}\cap\mathcal{H}_{\ell}\text{.}
\label{fifi}%
\end{equation}

\end{corollary}

\begin{proof}
We know from Theorem \ref{eigen1} above that $\widehat{H}$ and $\widetilde{H}$
have same eigenvalues and that $\Psi_{j,k}=W_{\psi_{k}}\psi_{j}$ satisfies the
eigenvalue equation $\widetilde{H}\Psi_{j,k}=\lambda_{j}\Psi_{j,k}$. Since
$\widehat{H}$ is self-adjoint and its eigenvalues are distinct, its
eigenfunctions $\psi_{j}$ form an orthonormal basis of $L^{2}(\mathbb{R}^{n}%
)$; it follows from Lemma \ref{Wong} that the $\Psi_{j,k}$ form an orthonormal
basis of $L^{2}(\mathbb{R}^{2n})$, hence there exist non-zero scalars
$\alpha_{j,k,\ell}$ such that $\Psi_{j}=\sum_{k,\ell}\alpha_{j,k,\ell}%
\Psi_{k,\ell}$. We have, by linearity and using the fact that $\widetilde
{H}\Psi_{k,\ell}=\lambda_{k}\Psi_{k,\ell}$,
\[
\widetilde{H}\Psi_{j}=\sum_{k,\ell}\alpha_{j,k,\ell}\widetilde{H}\Psi_{k,\ell
}=\sum_{k,\ell}\alpha_{j,k,\ell}\lambda_{k}\Psi_{k,\ell}.
\]
On the other hand we also have $\widetilde{H}\Psi_{j}=\lambda_{j}\Psi_{j}$,
\[
\widetilde{H}\Psi_{j}=\lambda_{j}\Psi_{j}=\sum_{j,k}\alpha_{j,k,\ell}%
\lambda_{j}\Psi_{k,\ell}%
\]
which is only possible if $\alpha_{j,k,\ell}=0$ for $k\neq j$; setting
$\alpha_{j,\ell}=\alpha_{j,j,\ell}$ formula(\ref{fifi}) follows. (That
$\Psi_{j,\ell}\in\mathcal{H}_{j}\cap\mathcal{H}_{\ell}$ is clear using the
definition of $\mathcal{H}_{\ell}$ and the sesquilinearity of the cross-Wigner transform.)
\end{proof}

We remark that the continuous spectrum can be dealt with in a similar fashion
provided that one generalizes the transform $W_{\phi}$ by allowing the
\textquotedblleft parameter\textquotedblright\ to be a tempered distribution
(in which case the normalization condition $||\phi||_{L^{2}}=1$ does no longer
make sense, of course); the same remark applies to the case where $\widehat
{H}$ is no longer essentially self-adjoint (cf. the remark following the proof
of Theorem \ref{eigen1}). To illustrate this, let us consider the two
following typical examples (in dimension $n=1$):

\begin{itemize}
\item $H(x,p)=p.$ In this case $\widehat{H}=-i\hbar\partial/\partial x$ is a
symmetric operator and the equation $\widehat{H}\psi=E\psi$ has solutions for
every real value of $E$; these solutions are the tempered distributions
$\psi(x)=C\exp(iEx/\hbar)$ ($C$ any complex constant). A straightforward
calculation shows that%
\[
W_{\phi}\psi(x,p)=C^{\prime}e^{\frac{2i}{\hbar}(E-p)x}\overline{F\phi(p)}%
\]
where $C^{\prime}$ is a new constant and $F\phi$ is the Fourier transform of
$\phi$. If we let $\phi$ range over $\mathcal{S}^{\prime}(\mathbb{R}^{n})$ and
use the fact that the Fourier transform is an automorphism of $\mathcal{S}%
^{\prime}(\mathbb{R}^{n})$ we see that $W_{\phi}\psi$ can be any distribution
of the type
\[
\Psi(x,p)=\Phi(p)e^{\frac{2i}{\hbar}(E-p)x}%
\]
with $\Phi\in\mathcal{S}^{\prime}(\mathbb{R}^{n})$; these distributions are
precisely the solutions of the stargenvalue equation
\[
p\star\Psi=(p-\tfrac{1}{2}i\hbar\partial_{x})\Psi=E\Psi
\]
as a straightforward calculation shows.

\item $H(x,p)=x.$ Here $\widehat{H}$ is the operator of multiplication by $x$;
this a symmetric operator without any eigenvalues and eigenfunctions. It is
however self-adjoint, and the solutions of $\widehat{H}\psi=E\psi$ are the
distributions $\psi=C\delta(x-E)$; one finds by an argument similar to that
above that $W_{\phi}\psi$ can be any distribution of the type
\[
\Psi(x,p)=\Phi(x)e^{-\frac{2i}{\hbar}(E-x)p}%
\]
which is the general solution of the stargenvalue equation%
\[
x\star\Psi=(x+\tfrac{1}{2}i\hbar\partial_{p})\Psi=E\Psi.
\]

\end{itemize}
The previous treatment of the stargenvlue equation for operators with a continuous spectrum 
can be made rigorous in the setting of Gelfand triples. In our setting $(\mathcal{S}(\mathbb{R}^n),L^2(\mathbb{R}^n),\mathcal{S}^\prime(\mathbb{R}^n))$ is the Gelfand triple of interest and the corresponding weak formulation of the eigenvalues and eigenvectors of an operator from $\mathcal{S}^\prime(\mathbb{R}^n)$ to $\mathcal{S}(\mathbb{R}^n)$. For further information on Gelfand triples we refer the reader to the standard reference \cite{gesh68}. 
\section{An Example and its Extension}

As an illustration consider the harmonic oscillator Hamiltonian%
\begin{equation}
H=\frac{1}{2}(p^{2}+x^{2}). \label{h1}%
\end{equation}
In view of the results above the spectra of the operators $\widehat{H}$ and
$\widetilde{H}$ are identical. Choosing for simplicity $\hbar=1$ the
eigenvalues of $\widehat{H}$ are the numbers $\lambda_{N}=N+\frac{1}{2}$ with
$N=0,1,2,...$. The normalized eigenfunctions are the rescaled Hermite
functions%
\begin{equation}
\psi_{k}(x)=(2^{k}k!\sqrt{\pi})^{-\frac{1}{2}}e^{-\frac{1}{2}x^{2}}%
\mathcal{H}_{k}(x). \label{h2}%
\end{equation}
where
\[
\mathcal{H}_{k}(x)=(-1)^{km}e^{x^{2}}\left(  \tfrac{d}{dx}\right)
^{k}e^{-x^{2}}%
\]
is the $k$-th Hermite polynomial. Using definition (\ref{wpf}) of $W_{\phi}$
together with known formulae for the cross-Wigner transform of Hermite
functions (see for instance Wong \cite{Wong}, Chapter 24, Theorem 24.1) one
finds that the eigenfunctions of $\widetilde{H}$ are linear superpositions of
the functions%
\begin{equation}
\Psi_{j+k,k}(z)=(-1)^{j}\left(  \tfrac{j!}{(j+k)!}\right)  ^{\frac{1}{2}%
}2^{\frac{k}{2}+1}\overline{\zeta}^{k}\mathcal{L}_{j}^{k}(2|z|^{2}%
)e^{-|\zeta|^{2}} \label{h3}%
\end{equation}
where $\zeta=x+ip$ and $\Psi_{j,j+k}=\overline{\Psi_{j+k,k}}$ for
$k=0,1,2,...$; here%
\[
\mathcal{L}_{j}^{k}(x)=\tfrac{1}{j!}x^{-k}e^{x}\left(  \tfrac{d}{dx}\right)
^{j}(e^{-x}x^{j+k})\text{ , }x>0
\]
is the Laguerre polynomial of degree $j\ $and order $k$. (For similar results
see Bayen et al. \cite{BFFLS2}.)

Notice that the example above can be generalized without difficulty to the
case of arbitrary quadratic Hamiltonians of the type
\[
H=\frac{1}{2}Mz\cdot z
\]
where $M$ is a positive-definite symmetric matrix. In fact, in view of
Williamson's diagonalization theorem there exists a symplectic matrix $S$ such
that%
\[
M=S^{T}DS\text{ \ , \ }D=%
\begin{pmatrix}
\Lambda & 0\\
0 & \Lambda
\end{pmatrix}
\]
where $\Lambda$ is the diagonal matrix whose entries are the moduli
$\omega_{j}>0$ of the eigenvalues $\pm i\omega_{j}$ of $JM$. We thus have%
\[
H\circ S=\sum_{j=1}^{n}\frac{\omega_{j}}{2}(x_{j}^{2}+p_{j}^{2})
\]
and $\widehat{H\circ S}=\widehat{S}\widehat{H}\widehat{S}^{-1}$ where
$\widehat{S}$ is anyone of the two metaplectic operators associated with $S$.
The eigenvalues of $\widehat{H\circ S}$ and $\widehat{H}$ are the same; they
are the numbers
\[
\lambda_{N_{1},...,N_{n}}=\sum_{j=1}^{n}\left(  N_{j}+\tfrac{1}{2}\right)
\omega_{j}%
\]
and then the eigenfunctions $\psi_{S}$ of $\widehat{H\circ S}$ and those,
$\psi$, of $\widehat{H}$ by the formula $\psi_{S}=\widehat{S}\psi$. Now, the
eigenfunctions of $\widehat{H}$ are tensor products of rescaled Hermite
functions; using the fact that $\psi_{S}=\widehat{S}\psi$ together with the
symplectic covariance formula
\[
W(\widehat{S}\psi,\widehat{S}\phi)(z)=W(\psi,\phi)(S^{-1}z)
\]
satisfied by the cross-Wigner distributions, the eigenvalues of $\widetilde
{H}=H\star$ are calculated in terms of tensor products of the functions
(\ref{h3}). We do not give the details of the calculations here since they are
rather lengthy but straightforward.

\section{Concluding Remarks}

Due to limitation of length there are several aspects of our approach we have
not discussed in this Letter. For instance, he methods developed here should
apply with a few modifications (but in a rather straightforward way) to more
general phase space (for instance co-adjoint orbits). A perhaps even more
exciting problem is the following, which is closely related to our previous
results \cite{golulett} on the relationship between the uncertainty principle
and the topological notion of symplectic capacity. A rather straightforward
extension of the methods we used in \cite{golulett} shows that if
\begin{equation}
|W_{\phi}\psi(z)|\leq Ce^{-\frac{1}{\hbar}(a|x|^{2}+b|y|^{2})}\text{ \ for
}z\in\mathbb{R}^{2n}%
\end{equation}
then we must have $ab\leq1$. In\ particular we can have $|W_{\phi}\psi(z)|\leq
Ce^{-\frac{1}{\varepsilon}|z|^{2}}$ only if $\varepsilon\geq\hbar$; it follows
that the Hilbert spaces $\mathcal{H}_{\phi}$ do not contain any nontrivial
function with compact support: assume in fact that $\Psi\in\mathcal{H}_{\phi}$
is such that $\Psi(z)=0$ for $|z|\geq R>0$. Then, given an arbitrary
$\varepsilon<\hbar$ one can find a constant $C_{\varepsilon}$ such that
$|\Psi(z)|\leq C_{\varepsilon}e^{-\frac{1}{\varepsilon}|z|^{2}}$, which is
impossible since $\Psi=W_{\phi}\psi$ for some $\psi\in L^{2}(\mathbb{R}^{n})$.
This suggests (taking Theorem \ref{eigen1} into account) that the solutions
$\Psi$ of the $\star$-genvalue equation cannot be too concentrated around a
point in phase-space. In fact we conjecture that if an estimate of the type
$|\Psi(z)|\leq Ce^{-\frac{1}{\hbar}Mz\cdot z}$ ($M$ symmetric
positive-definite) holds for an eigenfunction of the stargenvalue equation, then the symplectic capacity of the ellipsoid $Mz\cdot
z\leq\hbar$ must be at least $\frac{1}{2}h$. We will come back to this topic
in a near future.

\begin{acknowledgement}
The authors would like to express their deep gratitude to Prof. Daniel
Sternheimer for useful comments and valuable suggestions on an earlier draft
of this paper.
\end{acknowledgement}

\end{document}